\documentclass[runningheads]{llncs}
\usepackage{graphicx}

\usepackage{algorithm}
\usepackage{algorithmic}
\usepackage{booktabs}
\usepackage{amsmath}
\usepackage{amsthm}
\usepackage{tikz}
\usetikzlibrary{graphs}
\tikzstyle{post}=[->,shorten >=1pt, >=stealth,thin]
\usepackage{pgfplots}
\pgfplotsset{compat=1.18}
\theoremstyle{definition} \newtheorem{observation}[theorem]{Observation}
\usepackage{bibnames}

\begin{document}
\title{Cooperative Task Execution in \\Multi-Agent Systems }

%
\author{Karishma\inst{1}\orcidID{0000-0003-3842-7408} \and
Shrisha Rao\inst{2}\orcidID{0000-0003-0625-5103}}
%
%
\institute{International Institute of Information Technology, Bangalore, Karnataka 560100
\email{\{karishma,srao\}@iiitb.ac.in}}
\maketitle              

\begin{abstract}
We propose a multi-agent system that enables groups of agents to collaborate and work autonomously to execute tasks. Groups can work in a decentralized manner and can adapt to dynamic changes in the environment. Groups of agents solve assigned tasks by exploring the solution space cooperatively based on the highest reward first. The tasks have a dependency structure associated with them. We rigorously evaluated the performance of the system and the individual group performance using centralized and decentralized control approaches for task distribution. Based on the results, the centralized approach is more efficient for systems with a less-dependent system $G_{18}$ (a well-known program graph that contains $18$ nodes with few links), while the decentralized approach performs better for systems with a highly-dependent system $G_{40}$ (a program graph that contains $40$ highly interlinked nodes). We also evaluated task allocation to groups that do not have interdependence. Our findings reveal that there was significantly less difference in the number of tasks allocated to each group in a less-dependent system than in a highly-dependent one. The experimental results showed that a large number of small-size cooperative groups of agents unequivocally improved the system's performance compared to a small number of large-size cooperative groups of agents. Therefore, it is essential to identify the optimal group size for a system to enhance its performance.

\end{abstract}
%
%
\section{INTRODUCTION}


In a multi-agent system, a group is composed of individual agents who work collectively towards common goals or objectives. These agents may possess varying degrees of autonomy and can interact with each other and their environment to coordinate their actions. Groups are essential in multi-agent systems as they enable agents to collaborate, coordinate, and accomplish complex tasks or objectives that may be beyond the capabilities of individual agents. Some common approaches to group agents for task execution in a multi-agent system are hierarchical structure, task-oriented approaches, role-based approaches, learning-based approaches, cluster-based approaches, etc. 

Multi-agent systems involve multiple agents that work collaboratively to achieve specific goals within a group. MAS can be employed for dynamic task allocation within a group and overcome a distributed system's complex task allocation problem. Agents can negotiate for resource allocation based on the current resource availability~\cite{Shen_9732210}. Within a group, agents can work together to maximize resource utilization, guaranteeing effective use of computational resources like virtual machines, amazon EC2 instances, and storage~\cite{Gudu_7816893}.

Task allocation can be done using either a centralized or distributed approach~\cite{macarthur2011multi,zheng2008reaction}. Task scheduling is performed using the centralized approach without considering where tasks or requirements change over time~\cite{shehory1998methods,zheng2008reaction}. A centralized or distributed component must schedule the $m$ number of functions between the $n$ number of agents, where $m$ can be higher than $n$. In that case, multiple scheduling rounds are required, and each task will be accomplished on its scheduled turn~\cite{ramchurn2010coalition}. A few existing algorithms for task scheduling in a distributed system are SWARM-based approach~\cite{du2010task,ghassemi2019decentralized}, negotiation approach~\cite{luo2010multi,wang2013community}, and distributed constraint optimization problems~\cite{ramchurn2010decentralized}, etc.




Li et al.~\cite{Jiaoyang2021} presented a framework for solving multi-agent path finding, addressing path collisions. However, unlike our approach, their method does not require exploring solutions where multiple tasks may have similar outcomes.

To maximize system performance, Gerkey and Matarić ~\cite{Brian2004} survey solutions for multi-robot task allocation, taking several variables into account such as resource allocation, resource cost, coordination method, etc.  The work is specific to the allocation of spatial tasks to robots, unlike our model which presents an approach for solution space exploration in the context of agents that may not be robots.

Karishma and Rao~\cite{karishma2023cooperative} proposed an approach where a set of agents worked together to solve a set of tasks. During the solution space exploration, agents can ask for help from an oracle within a certain budget if needed. In situations where a system involves a large number of tasks and agents, the centralized entity responsible for task assignment, solution validation, and reward allocation can become a bottleneck for the entire system. We have extended this work to resolve this bottleneck by dividing the agents into multiple groups. 


In this work, we address the fundamental problem of solving distributed tasks by groups of agents. Agents explore a solution space to execute tasks. And agents collect the inference data (get the tasks for which same solution is applicable) in case of successful task validation. And they can also infer a new solution by using previously explored solutions. If a task is part of inference data, then the solution space exploration phase is not required. We have evaluated system performance with centralized and decentralized control for task assignment at the group level. Individual groups work in parallel to enhance multiple agent's efficiency and effectiveness in task execution. Tasks have dependencies within and across groups that should be executed first. We formulate and answer the following questions during experiments:
\begin{enumerate}
\item If there is a choice between centralized and decentralized control, which should be preferred and why?
\item Will the increase in speed of the agents for solution exploration increase the system performance?
\item Will the groups get tasks distributed equally in both less-dependent systems (LDS) and highly-dependent systems (HDS)?
\end{enumerate}

The simulation results show that dividing agents into smaller groups improves the system's performance  (see Table~\ref{tab:systemExplTimeG40chap4} and Table~\ref{tab:systemExplTimeG40VaringGroupschap4}). We have evaluated task allocation to groups that don't have inter-dependency, and we have observed that the tasks were almost evenly distributed for LDS but not in HDS (see Figure~\ref{fig:task_distribution_chap4}). $G_{18}$ and $G_{40}$ are well-known program graphs and widely used~\cite{Agrawal2014}. We have used $G_{18}$ and $G_{40}$ as examples of LDS and HDS respectively. Evaluation of centralized and decentralized control approaches shows that the centralized approach performs better for a system with less number of tasks, whereas the decentralized approach performs better for a large-scale system (see Figure~\ref{fig:comparision_centVSDecet_chap4} and Table~\ref{tab:group_WT_Centralized_and_Decentralized_chap4}). Evaluation of task distribution shows that LDS performs better when tasks assigned to a group are not dependent on the other group's task( independent set of tasks), whereas HDS performs better when the inter-dependency exists with the other group's tasks (see Table~\ref{tab:group_WT_Decentralized_chap4}).


The mathematical results prove the transitivity of knowledge within the group due to the sharing of gained knowledge between the agents (see Theorem~\ref{theorem1_Transitivity}). It also formulates the expected waiting time $E[W]$ due to dependencies between the tasks, which equals $\Theta(mkp^k)$ (see Theorem~\ref{theorem_dep_impact}). This is in line with the results, which are shown in Table~\ref{tab:group_WT_Decentralized_chap4}. Mathematical results to identify the optimal group size show that system performance is better with the small size of a large number of groups over the large size of a small number of groups based on the expected system execution time (see Theorem~\ref{theorem_grp_size} and Table~\ref{tab:systemExplTimeG40chap4}).

The rest of the paper is structured as follows. Section~\ref{CES} provides the details about the system model, and it also explains the cooperative execution strategy for groups of agents. Section~\ref{results} presents the experimental results obtained through simulation, and it also describes the mathematical results. 

\section{COOPERATIVE EXECUTION STRATEGY FOR GROUPS OF AGENTS} \label{CES}
We present a model for a multi-agent system, which has a set of cooperative agents working on inter-dependent tasks to explore solution space and execute the tasks. We are paving the way for an efficient and effective system design by evaluating both centralized and decentralized task allocation approaches.

\subsection{System Model}\label{sm}
We present a multi-agent system model that has requirements to execute and implement a set of tasks by a set of groups. Tasks have some dependency structure among them. A group consists of cooperative agents who actively share their knowledge within the group.
\begin{figure}
  \centering
  \includegraphics[scale=0.30]{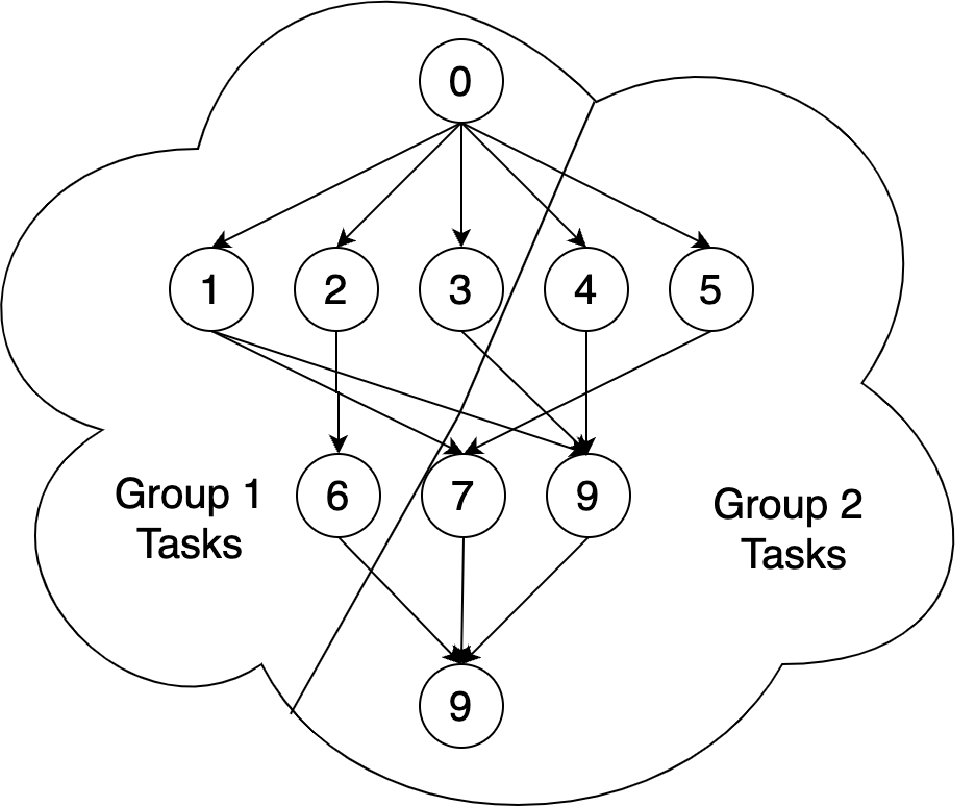}
   \caption{\label{fig:tasksForGroups_G10} Task Dependency Graph $G_{10}$ for groups.}
   \vspace{-5pt}
\end{figure}

We have a set of tasks, each with a fixed reward and a set of dependent tasks. We divide the agents into multiple groups to distribute the workload among the groups instead of controlling from a centralized entity. We partition set of tasks into multiple disjoint subsets for different groups of agents, allowing them to be executed independently. For this, we split the rewards and dependencies accordingly. We make sure that no duplicity of tasks exists among the groups. A group $g_i$ can get a set of tasks that may be dependent on tasks that belong to group $g_j$, similar to Figure~\ref{fig:tasksForGroups_G10}. A program graph $G_{10}$~\cite{Bian_10233785} consists of $10$ tasks and illustrates the task distribution to two different groups and has task inter-dependencies. A task can only be scheduled for execution once all the tasks in its dependency list are executed. We prioritize the tasks based on the reward factor. Agents within the group share the explored solution and inference data within the group but not outside the group. The advantage of not sharing knowledge outside the group is that it reduces the communication overhead.


We consider a standard system model of $n$ agents, $l$ groups, and $m$ tasks. We'll represent the information about tasks, solutions, groups, and the agent's knowledge using a mathematical model. Let:

\begin{itemize}

\item \(\mathcal{A}\), \(\mathcal{G}\), \(\mathcal{T}\), and \(\mathcal{S}\) are the sets
  of agents, groups, tasks, and solutions, respectively.
\item \(a_i \in \mathcal{A}\) denotes the \(i\)-th agent.
\item \(g_k \in \mathcal{G}\) denotes the \(k\)-th group.
\item \(t_j \in \mathcal{T}\) represents the \(j\)-th task.
\item \(s_j \in \mathcal{S} \) denotes the solution corresponding to task \(t_j\).

\end{itemize}

We have a set of rewards $\mathcal{R}$ and a set of dependencies $\mathcal{D}$ associated with our task set $\mathcal{T}$, where $d_j \subseteq T \setminus \{t_j\}$ consists of a set of tasks on which $t_j$ is dependent. To explore the solution space, we have divided $n$ agents into $l$ groups, each containing $n/l$ agents. $\mathcal{A}_{subset}$ represents a subset of agents that is allocated to group $g_k$. We have created $l$ subsets of the task set $\mathcal{T}$, which we represent as $\mathcal{T}_{subset}$. We have also created respective subsets of rewards $\mathcal{R}_{subset}$ and dependencies $\mathcal{D}_{subset}$.

Let $\mathcal{S}$ represent the solution space, which is defined as the set of all possible solutions $s$ for the given set of tasks $\mathcal{T}$. Formally:
\begin{equation*} \mathcal{S} = {s_j \mid s_j \text{ is a valid solution for some task } t_j \in \mathcal{T}} \end{equation*}
In other words, for each task $t_j$ in the set of tasks $\mathcal{T}$, there exists a set of valid solutions $s_j$. The overall solution space $\mathcal{S}$ is the union of all these individual solution sets:
\begin{equation*} \mathcal{S} = \bigcup_{t_j \in \mathcal{T}} s_j \end{equation*}

Intuitively, the solution space $\mathcal{S}$ encompasses all potential solutions that the agents can explore for the given tasks. Each solution $s_j$ within this space corresponds to a specific task $t_j$. The agents navigate this solution space searching for valid solutions to execute the tasks.

A task assignment function $\lambda$ assigns a subset of tasks from $\mathcal{T}_{subset}$ to group $g_k$. It is denoted as $\lambda(g_k):\mathcal{A}\rightarrow2^{\mathcal{T}_{subset}}$. The function ensures that task assignments to different groups are non-overlapping, meaning that $\lambda(g_k) \cap \lambda(g_h) = \emptyset$ if $k \neq h$. $\mu(g_k)$ represents the set of tasks accomplished by $g_k$ such that $\mu(g_k) \subseteq \lambda(g_k)$.We have used a set \(K(a_i)\) to represent the knowledge of an agent \(a_i\), where \(K(a_i)\) contains \(\{(t_j, s_j) \mid t_j\ \in \mathcal{T}\) and \(s_j \in \mathcal{S}\)\}.

Every task \(t_j\) has a corresponding solution \(s_j\). If an agent \(a_i\) knows \(s_j\) for some task \(t_j\), then \(a_i\) also knows the inference data of \(t_j\). If \(t_k\) belongs to the inference data of \(t_j\), then the same solution \(s_j\) is valid for task \(t_k\) as well, which is indicated by \(t_k \sim t_j\).

\subsection{CES Algorithm} \label{CES_algo}

This section explains the strategy for exploring the solution space by groups of agents. In a distributed system, all agents are divided into different groups to work on a set of tasks. Each set of tasks is further divided into subsets, and each group of agents is assigned a specific subset of tasks to work on. The agents in each group work together to explore solutions for their assigned tasks. Once they have found a solution, they execute the task. Coordination among agents in the same group is often necessary to share the gained knowledge during the task execution process.

Algorithm~\ref{alg:centralizedControlAtIndividualGroupLevelChap4} describes the solution space exploration at the group level. In this approach, a group is assigned a set of agents who will execute a set of tasks by considering the respective rewards and dependencies. The group has a centralized control that takes care of task assignment and validation. Later, we enhanced the system by adding the decentralized task distribution where agents pull tasks from the set of available tasks within the group.

\begin{algorithm}[tbh]
\caption{Solution Space Exploration at Group Level Algorithm.}
\label{alg:centralizedControlAtIndividualGroupLevelChap4}
\textbf{Input}: $\mathcal{T}_{subset}$: A subset of tasks, $\mathcal{R}_{subset}$: A subset of respective rewards, $\mathcal{D}_{subset}$: A subset of respective dependencies, $\mathcal{A}_{subset}$: A subset of agents\\
\textbf{Output}: Share knowledge with all the agents within the group
\begin{algorithmic}[1] 
\STATE $\mathcal{A}_{avail}$ $ \gets getAvailAgents()$ \label{centralizedControlAtIndividualGroupLevelChap4_marker}
\STATE $\mathcal{T}_{avail} \gets \mathit{getAvailTasks(\mathcal{T}_{subset}, \mathcal{R}_{subset}, \mathcal{D}_{subset}, \mathcal{A}_{avail} )}$ 
\STATE{// Assign the tasks to available agents within the group}
\STATE $\mathit{taskAssignment(\mathcal{T}_{avail}, \mathcal{A}_{avail})}$
\WHILE{true}
\STATE {// On receive event listener for solution validation from an agent $a_{i}$}
\STATE $\mathit{validateSolution(t_j, s_j)}$
\STATE{// Allocate the reward based on the validation result of the solution for $t_j$}
\STATE $ \mathit{allocateReward( a_{i}, \mathcal{T}_{subset}, \mathcal{R}_{subset} )} $
\IF{isRewarded}
\STATE{// Remove the dependencies from the dependent task on $t_j$}
\STATE $\mathit{updateDependencies(t_{j})}$ 
\FOR{\textbf{each} $a_i \in \mathcal{A}_{subset}$}
\STATE $shareKnowledge(t_j, s_j, \mathcal{A}_{subset})$
\ENDFOR
\ENDIF
\STATE \textbf{go to} \ref{centralizedControlAtIndividualGroupLevelChap4_marker}
\ENDWHILE
\end{algorithmic}
\end{algorithm}

Algorithm~\ref{alg:centralizedControlAtIndividualGroupLevelChap4} accepts a subset of tasks $\mathcal{T}_{subset}$ with respective dependencies $\mathcal{D}_{subset}$ and rewards $\mathcal{R}_{subset}$. Solution space exploration is performed by $\mathcal{A}_{subset}$ agents. As a result, the knowledge gained by all the $\mathcal{A}_{subset}$ agents is shared among them. In algorithm~\ref{alg:centralizedControlAtIndividualGroupLevelChap4}, in line 1, get the available agents to execute the tasks. Initially, all the agents are available, but it's possible that a few agents are working on solution space exploration in the next iteration. In line 2, get the available tasks for the available agents. In line 4, it assigns unique tasks to available agents. If the available tasks are less than the available agents, then a few agents do not get any task assigned. In line 7, $\mathit{validateSolution(t_j, s_j)}$ validates the explored solution $s_j$ for a task $t_j$. In line 9, $\mathit{allocateReward( a_{i}, \mathcal{T}_{subset}, \mathcal{R}_{subset} )}$ allocates a reward for the explored solution $s_j$ if it is valid; otherwise, there is no reward for task $t_j$. In line 12, dependencies on the task $t_j$ are removed from the dependency set for all the remaining tasks only if the reward is allocated for the task $t_j$. In line 14, $shareKnowledge(t_j, s_j, \mathcal{A}_{subset})$ shares the gained knowledge by an agent $a_i$ to all the agents who are part of set $\mathcal{A}_{subset}$.

\section{RESULTS}\label{results}

We present mathematical and simulation results for a cooperative execution strategy to execute tasks by groups of agents.

\subsection{Mathematical Results}\label{mr}

Here are some essential mathematical results for gained knowledge, comparing the impact of various sizes of the groups and evaluating the expected waiting time of a task when the dependency graph is associated among the tasks.

\begin{theorem}[Transitivity on Knowledge] \label{theorem1_Transitivity}If \(t_k \sim t_l\), agent \(a_i\) knows the solution for task \(t_k\):\((t_k, s_k) \in  K(a_i)\) and two agents, \(a_i\) and \(a_j\) belong to the same group: \({a_i, a_j} \in g_p\), then \((t_l, s_k) \in  K(a_j)\)
\end{theorem}

\begin{observation} [Impact of Task Dependencies]\label{observation_task_dependencies}
Task distribution is uneven in both less-dependent systems (LDS) and highly-dependent systems (HDS) (see Figure~\ref{fig:task_distribution_chap4}).
\end{observation}

\begin{theorem} [Impact of Dependencies on Expected Waiting Time for LDS] \label{theorem_dep_impact}
Consider the same multi-agent system as defined previously with a set of $m$ tasks $\mathcal{T}$ and agents $\mathcal{A}$. Let the dependency graph among tasks have maximum degree $k < m-1$, so each task depends on at most $k$ other tasks. Let $p$ be the probability a dependency is unresolved. Then, the expected waiting time scales as $E[W] = \Theta(mkp^k)$.  
\end{theorem}

\begin{proof}
For each task $t_i$, define a binary random variable $X_i$: 
\[ X_i = \begin{cases} 
      1 & \text{if $t_i$ has unresolved dependencies}\\
      0 & \text{otherwise}
   \end{cases}
\]

And the total waiting time is $W = \sum_{i=1}^{m} X_i$. 

Since the maximum degree is $k$, each task has at most $k$ dependencies. By the law of total probability:

\[E[X_i] = 1 - (1-p)^k\]

Therefore, the expected total waiting time is:
\begin{align*} 
E[W] &= \sum_{i=1}^{m} E[X_i] \\
       &= m(1 - e^{k\ln(1-p)})
\end{align*}

Using the Taylor approximation $e^x \approx 1+x$ for small $x$: 
\begin{align*}
E[W] &\approx m(1 - (1+k\ln(1-p))) \\
       &= mk(-p)^k \\
       &= \Theta(mkp^k)
\end{align*}
Thus, the waiting time scales as $\Theta(mkp^k)$ when the maximum dependency degree is $k$.
\end{proof}

\begin{theorem} [Impact of Dependencies on Expected Waiting Time for Fully Connected Graph ]\label{theorem_dep_impact_HDS}
Consider a multi-agent system with a set of agents $\mathcal{A} = \{a_1, \dotsc, a_n\}$ and a set of tasks $\mathcal{T} = \{t_1, \dotsc, t_m\}$ where $|\mathcal{T}|=m$. Let the dependency graph among tasks be fully connected, such that each task $t_i$ depends on all other tasks $t_j$ where $j \neq i$. Let $d = m-1$ be the number of dependencies per task. Further, assume that the probability of any dependency being unresolved is a constant $p \in (0,1)$. Then, the expected waiting time $E[W]$ for an agent to receive an executable task scales as $\Theta(m p^d)$.
\end{theorem}

Table~\ref{tab:group_WT_Decentralized_chap4} presents the total waiting time caused by task dependencies for two systems: one with low dependency and one with high dependency, which is in line with Theorem~\ref{theorem_dep_impact} and Theorem~\ref{theorem_dep_impact_HDS}

\begin{theorem} [Optimal Group Size] \label{theorem_grp_size}
Consider a multi-agent system with $n$ agents of fixed capability, partitioned into $l$ groups, exploring a set of $m$ independent tasks. Let $T(g_k)$ be the random variable denoting the time taken by group $g_k$ to complete its assigned tasks. If the number of groups $l$ increases while keeping $n$ and $m$ fixed, thereby decreasing the group size, then the expected system execution time $E[\max_k T(g_k)]$ decreases.
\end{theorem}

\begin{proof}
With $m$ independent tasks split evenly between groups, each group gets $m/l$ tasks. Since agents have fixed capabilities, the group completion time $T(g_k)$ is approximately normally distributed according to the central limit theorem, with $E[T(g_k)] = (m/l)/v$ where $v$ is the fixed agent capability parameter.

Additionally, $\text{Var}(T(g_k)) = \frac{\sigma^2}{l_k}$ where $l_k = \frac{n}{l}$ is the group size, and $\sigma^2$ measures variability inherent to the tasks and environment.

Since \(\max(X_1, \dotsc, X_n) \leq X_1 + \dotsc + X_n\), we have: 

\begin{align*}
E\left[\max_k T(g_k)\right] &\leq E\left[\sum_k T(g_k)\right] = lE\left[T(g_k)\right] = \frac{m}{v}
\end{align*}

This expected max group time is constant with respect to changes in $l$. However, increasing groups $l$ reduces group size $l_k$, thereby increasing the variance Var$(T(g_k))$. By properties of distributions of maxima:

$$E\left[\max_k X_k\right] \leq E\left[\max_k Y_k\right] \text{ if } X_k \leq_{st} Y_k \:\: \forall k$$

Smaller groups have a higher variance in completion times. Therefore, moving from fewer groups/larger groups to more groups/smaller groups reduces $E[\max_k T(g_k)]$, the expected system execution time. This demonstrates that smaller groups improve expected performance.
\end{proof}

In our case, increasing the number of groups $l$ shrinks the group size $l_k$, which in turn increases the variance Var$(T(g_k))$ of each group's completion time. Higher variance indicates the distribution is more spread out, meaning group completion time is stochastically greater with smaller groups. Therefore, smaller groups reduce the expected system execution time. Experimental results in Table~\ref{tab:systemExplTimeG40chap4} confirm that system performance is better with a large number of small-size groups instead of a small number of large-size groups.

\subsection{Experimental Results}

The effectiveness of cooperative solutions for the execution of tasks is being tested across different scenarios. These include the distribution of tasks among multiple groups, the time taken by groups to explore solutions, the time taken by the system to explore solutions, variation in the speed of agents, adopting different approaches to task assignment among the agent groups, and the impact of HDS $G_{40}$, and LDS $G_{18}$. We conduct experiments where we generate a random maze with a random target location and vary the maze size. Multiple groups explore solutions on a maximum maze size of $400 \times 400$ in parallel, and our designed model can handle task dependencies to simulate real-time scenarios.

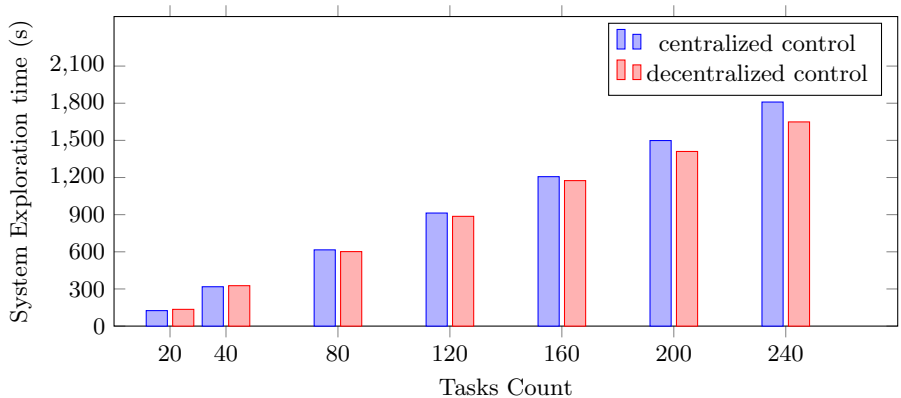
\begin{figure}
\centering
\begin{tikzpicture}
\begin{axis}[
  xlabel={Tasks Count},
  ylabel={System Exploration time (s)},
  ylabel near ticks,
  xmin=0, xmax=280,
  ymin=0, ymax=2500,
  ybar, ymin=0,
  bar width=8pt,
  height=5.7cm,width=12cm,
  xtick={20, 40, 80, 120, 160, 200, 240},
  ytick={0,300,600,900,1200,1500, 1800, 2100},
  ]
  \addplot coordinates { (20, 125.08) (40, 317.9) (80,615.78) (120, 913.07) (160,1206.91) (200,1498.54) (240,1809.6)};
  \addplot coordinates { (20, 135.23) (40, 326.45) (80,601.65) (120,886.43) (160, 1175.03) (200,1410.4) (240,1649.05)};
\legend{centralized control, decentralized control}
\end{axis}
\end{tikzpicture}
 \caption{System performance for centralized and decentralized group approach.}
 \label{fig:comparision_centVSDecet_chap4}
\end{figure}

The graph presented in Figure~\ref{fig:comparision_centVSDecet_chap4} compares two approaches for system performance. The first approach involves a centralized entity that assigns tasks to all agents within the group, while in the second approach, an individual agent selects their own tasks from a set of available tasks. The graph indicates that the first approach is more efficient for a small number of tasks, but the second approach outperforms it for a larger number of tasks in terms of execution time.

\begin{table}[h]
\centering
\caption{Varying number of groups ($5$ agents per group) to explore $500$ tasks.}
{\begin{tabular}{ccc}
\toprule
\textbf{Groups} & \pmb{$ET(g_k)$} & \pmb{$ET$} \\
\midrule
1       & 2458.30  & 2458.50  \\
2       & 1631.64  & 1705.09  \\
4       & 859.50   & 924.36   \\
6       & 635.66   & 686.5    \\
8       & 487.14   & 534.94   \\
10      & 368.32   & 409.46   \\
\bottomrule
\end{tabular}}
\label{tab:systemExplTimeG40chap4}
\end{table}

We have tested the performance of our model designed for groups by varying the number of groups from $1$ to $10$. During the experiment, shown in Table~\ref{tab:systemExplTimeG40chap4}, each group consisted of $5$ agents while maintaining a constant number of tasks at $500$. Each group was assigned tasks that depended on the tasks of the other groups. $ET(g_k)$ denotes the average execution time in seconds taken by a group $g_k$ and $ET$ denotes the overall system's execution time. Our results indicate that increasing groups improves the system's execution time. If we divide tasks among the groups, then it adds the complexity of dividing the tasks among groups and then collecting results. Our result confirms the system's stability in the case of many groups, which does not reduce the system's performance. The system's execution time is always higher than the $ET(g_k)$ because of several additional jobs at the system level, like splitting the tasks among the groups and collecting the results at the end.

\begin{table}[h]
\centering
\caption{Divide $50$ agents into groups to explore $500$ tasks.}
{\begin{tabular}{cc}
\toprule
 \textbf{Groups}& \pmb{$ET$} \\
\midrule
1       &    496.00   \\
2       &    483.70   \\
4       &    467.56   \\
6       &    451.48   \\
8       &    430.02   \\
10      &    409.46   \\
\bottomrule
\end{tabular}}
\label{tab:systemExplTimeG40VaringGroupschap4}
\end{table}

Table~\ref{tab:systemExplTimeG40VaringGroupschap4} presents the system's performance with a fixed number of $50$ agents and an increasing number of groups, ranging from $1$ to $10$. We have distributed the agents almost equally among the groups to explore the solution space for the tasks. The results indicate that although the system consists of $50$ agents, dividing them into smaller groups enhances the overall system performance. There are various factors that affect the system performance in this case, like task distribution among groups and knowledge transitivity, as described in Theorem~\ref{theorem1_Transitivity}.

\begin{table}[h!]
\centering
\caption{Individual group performance and waiting time in a system with centralized and decentralized control at the group level.}
{\begin{tabular}{c ccc ccc }
\toprule
	& \multicolumn{3}{c}{\textbf{Centralized Control}}
	& \multicolumn{3}{c}{\textbf{Decentralized Control}} \\[0.1cm]
\cline{2-7} \\[-0.2cm]
\textbf{Group } & $ \pmb{ET(g_{k})}$ & $ \pmb{|\lambda(g_{k})|} $ & $ \pmb{TWT(g_k)}$& $ \pmb{ET(g_{k})}$ & $ \pmb{|\lambda(g_{k})|} $ & $ \pmb{TWT(g_k)}$ \\[0.1cm]
\midrule
1       & 132.39  & 18 & 34.06      & 125.0  & 18 & 14.43   \\
2       & 128.41  & 18 & 30.34      & 123.26  & 18 & 15.07  \\
3       & 139.02  & 18 & 35.09      & 122.92  & 18 & 11.15  \\
4       & 128.89  & 18 & 27.23      & 121.56  & 18 & 15.84  \\
5       & 141.14  & 18 & 29.67      & 126.04  & 18 & 18.38  \\
\bottomrule
\end{tabular}}
\label{tab:group_WT_Centralized_and_Decentralized_chap4}
\end{table}

Table~\ref{tab:group_WT_Centralized_and_Decentralized_chap4} displays the performance of each group and the total waiting time of agents at the group level. This experiment is conducted on a system with a dependency like the $G_{18}$ program graph and $400\times 400$ maze size. In the centralized control approach, a centralized entity assigns tasks to each group, whereas, in the decentralized approach, each agent pulls the task from the available tasks within the same group. The total waiting time srao($TWT$) of a group $g_k$ is denoted as $TWT(g_k)$, is calculated by adding the waiting time of all agents who belong to the same group. The result shows that $ET(g_{k})$ and $TWT(g_k)$ for the decentralized control approach are better than the centralized control approach.

\begin{table}
\setlength{\tabcolsep}{12pt}
\centering
\caption{Individual group performance and waiting time in a system with decentralized control at group level.}
{\begin{tabular}{cccc}
\toprule
\textbf{System} & \textbf{Dependency}   & $\pmb{ET}$  & $\pmb{TWT}$   \\
\midrule
\textbf{less-}      & inter-dependency  & 631.67                       & 61.56                    \\
\textbf{dependent}  &                   & 628.41                       & 62.18                    \\
                    &                   & 629.81                       & 60.84                    \\
\cline{2-4}
                    & independency      & 619.49                       & 37.85                    \\
                    &                   & 613.88                       & 34.05                    \\
                    &                   & 601.24                       & 32.56                    \\
\cline{1-4}
\textbf{highly-}    & inter-dependency  & 710.80                       & 98.67                    \\
\textbf{dependent}  &                   & 718.15                       & 93.97                    \\
                    &                   & 709.65                       & 96.08                    \\
\cline{2-4}
                    &independency       & 725.80                       & 84.52                    \\
                    &                   & 732.65                       & 76.50                    \\
                    &                   & 728.34                       & 81.53                    \\

\bottomrule
\end{tabular}}
\label{tab:group_WT_Decentralized_chap4}
\end{table}

As per the design, tasks distributed among the groups have inter-dependency, and that can increase the waiting time of an agent to get a task to execute. So, we conducted an experiment on two sets of groups where the first set of groups had a set of tasks without task-dependency across the groups. In contrast, the second set of groups had inter-dependency among the tasks across the groups. Each group had $80$ tasks to implement. Table~\ref{tab:group_WT_Decentralized_chap4} shows the system's execution time and total waiting time ($TWT$) for both LDS and HDS. Total waiting time is the sum of the waiting time of all the agents at all the groups in the system. The result shows that both total waiting time and system execution are always less when comparing inter-dependency and independent tasks among the groups for LDS. For HDS, total waiting time is less, but system execution time is more when comparing inter-dependency and independent tasks among the groups. This is because of the additional step to identify the subsets of tasks that are not dependent on another subset of the tasks but can be dependent on another task that belongs to the same group.

We have done a comparison of the execution time taken by three different groups whose agents have dissimilarities in the speed to explore the solution space. Obtained results suggest that increasing the speed of agents within a group does not necessarily improve system performance linearly due to the dependencies on the other tasks.

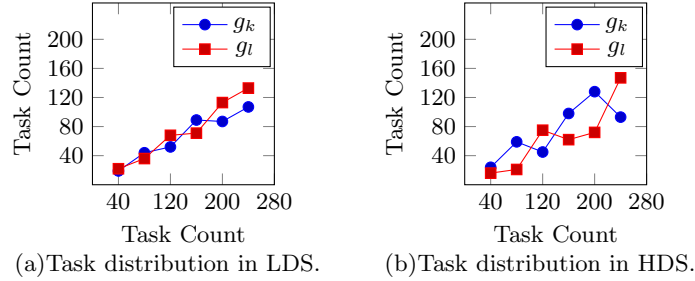
\begin{figure}
\centering
\begin{tikzpicture}
\begin{axis}[
xlabel={Task Count},
  ylabel={Task Count},
  ylabel near ticks,
  xmin=0, xmax=280,
  ymin=0, ymax=250,
  height=4cm,width=4cm,
  xtick={40,  120,  200, 280},
  ytick={40, 80, 120, 160, 200},
  ]
  \addplot coordinates { (40, 19) (80, 44) (120, 52) (160, 89) (200, 87) (240, 107)};
  \addplot coordinates { (40, 22) (80, 36) (120, 68) (160, 71) (200, 113) (240,133)};
\legend{$g_k$,$g_l$}
\end{axis}
 \label{fig:ld_task_distribution_chap4}
\node[below] at (1,-0.8) {(a)Task distribution in LDS.};
\end{tikzpicture}%
\hspace{15pt}
\begin{tikzpicture}
\begin{axis}[
  xlabel={Task Count},
  ylabel={Task Count},
  ylabel near ticks,
  xmin=0, xmax=280,
  ymin=0, ymax=250,
  height=4cm,width=4cm,
  xtick={40, 120, 200, 280},
  ytick={40, 80, 120, 160, 200},
  ]
  \addplot coordinates { (40, 24) (80, 59) (120, 45) (160, 98) (200, 128) (240, 93)};
  \addplot coordinates { (40, 16) (80, 21) (120, 75) (160, 62) (200, 72) (240,147)};
\legend{$g_k$,$g_l$}
\end{axis}
 \label{fig:hd_task_distribution_chap4}
\node[below] at (1,-0.8) {(b)Task distribution in HDS.};
\end{tikzpicture}
 \caption{Task distribution between groups of agents in LDS and HDS.}
 \label{fig:task_distribution_chap4}
\end{figure}

We allocate the subset of tasks to different groups without any inter-dependency. Figure~\ref{fig:task_distribution_chap4}(a) shows the number of tasks allocated to groups $g_k$ and $g_l$ for a LDS. Figure~\ref{fig:task_distribution_chap4}(b) shows the number of tasks allocated to groups $g_k$ and $g_l$ for a HDS. Figure~\ref{fig:task_distribution_chap4}(a) and Figure~\ref{fig:task_distribution_chap4}(b) show that the difference between the number of tasks allocated to two different groups in a LDS is less when compared with a HDS. Experimental results conclude that:
\begin{enumerate}
\item It is better to use a centralized control approach when the number of tasks is small in the system; a decentralized control approach is preferred when the number of tasks is huge (Figure~\ref{fig:comparision_centVSDecet_chap4}).
\item A large number of small-size cooperative groups of agents improves the system's performance when compared with a small number of large-size cooperative groups of agents (Table~\ref{tab:systemExplTimeG40chap4}).
\item Increasing the speed of agents in the groups improves the system performance up to a certain point due to inter-dependencies on the other group's tasks.
\item Due to the dependency, tasks are not evenly distributed for both LDS and HDS (Figure~\ref{fig:task_distribution_chap4}).
\item System performance is better in LDS when groups get an independent set of tasks (no task dependency on other group's tasks), whereas system performance in HDS is better when groups have an inter-dependency of tasks among groups (Table~\ref{tab:group_WT_Decentralized_chap4}).
\end{enumerate}

\section{CONCLUSIONS}\label{conclusion}

After dividing the agents into multiple groups, we investigated the system performance and distributed several jobs, like task assignment, solution validation, reward allocation, etc., to groups. We evaluated the system performance and individual group performance with the centralized and decentralized control approaches for task distribution. In this case, agents share knowledge within the respective group, which reduces the communication overhead. We have also evaluated task allocation to groups that don't have interdependence, and we have observed that the difference in the number of tasks allocated to each group is less in a LDS compared with a HDS. Varying group size analysis shows that a large number of small-size groups performs better when compared with a small number of large-size groups. This result will be beneficial when the system has a requirement to identify the optimal group size.











\end{document}